\documentclass[conference]{IEEEtran}
\IEEEoverridecommandlockouts
\usepackage{cite}
\usepackage{amsmath,amssymb,amsfonts,amsthm}
\theoremstyle{remark}

\newtheorem{prop}{Proposition}
\newtheorem{defi}{Definition}
\usepackage{algorithm}
\usepackage{algorithmic}
\usepackage{graphicx}
\usepackage{textcomp}
\usepackage{xcolor}
\usepackage{color}
\usepackage{stfloats}
\usepackage{subfig}
\begin{document}
\title{Energy-efficiency of Massive Random Access with Individual Codebook\\}
\author{\IEEEauthorblockN{Junyuan Gao, Yongpeng Wu, and Wenjun Zhang}
\thanks{J. Gao, Y. Wu, and W. Zhang are with the Department of Electronic Engineering, Shanghai Jiao Tong University, Minhang 200240, China (e-mail: sunflower0515@sjtu.edu.cn; yongpeng.wu@sjtu.edu.cn; zhangwenjun@sjtu.edu.cn) (Corresponding author: Yongpeng Wu). The work of Y. Wu is supported in part by the National Key R\&D
Program of China under Grant 2018YFB1801102,  JiangXi Key R\&D Program
under Grant 20181ACE50028, National Science Foundation
(NSFC) under Grant 61701301, the open research project of State Key Laboratory of Integrated
Services Networks (Xidian University) under Grant ISN20-03, and Shanghai Key Laboratory of Digital Media Processing and Transmission (STCSM18DZ2270700).
The work of W. Zhang is supported by  Shanghai Key Laboratory of Digital Media Processing and Transmission (STCSM18DZ2270700).}
}
\maketitle

\begin{abstract}
  The massive machine-type communication has been one of the most representative services for future wireless networks.
  It aims to support massive connectivity of user equipments (UEs) which sporadically transmit packets with small size.
  In this work, we assume the number of UEs grows linearly and unboundedly with blocklength and each UE has an individual codebook.
  Among all UEs, an unknown subset of UEs are active and transmit a fixed number of data bits to a base station over a shared-spectrum radio link.
  Under these settings, we derive the achievability and converse bounds on the minimum energy-per-bit for reliable random access over quasi-static fading channels with and without channel state information (CSI) at the receiver.
  These bounds provide energy-efficiency guidance for new schemes suited for massive random access.
  Simulation results indicate that the orthogonalization scheme TDMA is energy-inefficient for large values of UE density $\mu$.
  Besides, the multi-user interference can be perfectly cancelled when $\mu$ is below a critical threshold.
  In the case of no-CSI, the energy-per-bit for random access is only a bit more than that with the knowledge UE activity.
\end{abstract}

\section{Introduction} \label{section1}
  Driven by many use cases, such as Internet of Things, massive machine-type communication (mMTC) has been regarded as a necessary service in future wireless networks \cite{wuyp}.
  It aims to achieve the communication between massive user equipments (UEs) and the base station (BS), where only a fraction of UEs are active at any given time interval and transmit data payloads with small size.
  In order to reduce latency, the grant-free random access scheme is usually adopted, where the UE activity is unknown in advance for the receiver.

  Some subsets of these topics have been discussed in the past.
  When the number of UEs $K$ is finite and blocklength $n$ is infinite, the fundamental limits can be obtained based on classical multiuser information theory \cite{elements_IT}.
  Motivated by emerging systems with massive UEs, a new paradigm was proposed in \cite{GuoDN}, where $K$ was allowed to grow unboundedly with $n$ and the payload was infinite.
  The linear scaling was adopted in \cite{GuoDN, improved_bound, finite_payloads_fading, A_perspective_on, RAC_fading}.
  Since the data payload of UEs was usually of small size, \cite{improved_bound} and \cite{finite_payloads_fading} derived  bounds on the minimum energy-per-bit for reliable transmission with finite payload and infinite blocklength in AWGN and fading channels, respectively.
  \cite{A_perspective_on} and \cite{RAC_fading} studied random access limits in AWGN and fading channels, respectively, where common codebook was adopted with finite payload, blocklength, and active UEs. When common codebook is utilized, there is no obvious difference for theoretical derivation between random access and massive access with the knowledge of UE activity.

  In this work, we adopt individual codebook\footnote{It should be noted that individual codebook and common codebook assumptions correspond to different massive access models in practice, which jointly constitute a complete mMTC research [1].} for each UE. The number of UEs $K$ grows linearly and unboundedly with blocklength. Any $K_a$ of them have a fixed number of data bits to send over quasi-static fading channels. We assume synchronous transmission.
  This system model has not been studied before but is significant in mMTC.
  The infinite payload assumption in classical multiuser information theory will result in infinite energy-per-bit when $K$ and $n$ go to infinity with a fixed rate, which is not suitable for practice mMTC systems \cite{wuyp}. Thus, we consider finite payload and finite energy codewords in this paper for energy-efficiency \cite{wuyp}.
  We obtain the bounds on the minimum energy-per-bit for reliable random access with no channel state information (CSI) and CSI at the receiver (CSIR).
  These bounds provide energy-efficiency guidance for new massive random access coding and communication schemes.
  Since we do not know UE activity in advance, the derivation is more difficult compared with  \cite{finite_payloads_fading}.

  \emph{Notation:}
  We adopt uppercase and lowercase boldface letters to denote matrices and column vectors, respectively.
  $\mathbf{I}_{n}$ denotes the $n\times n$ identity matrix.
  Let $\left(\cdot \right)^{H}\!$, $\oplus$, $\left\|\cdot \right\|_{p}$, $\cdot\backslash\cdot$, $\left| \mathcal{A} \right|$, and $\operatorname{span}(\cdot)$ denote conjugate transpose, direct sum, ${\ell}_p$-norm, set subtraction, the cardinal of set $\mathcal{A}$, and the span of a set of vectors, respectively.
  Let $\mathcal{CN}(\cdot,\cdot)$, $\beta(\cdot,\cdot)$, $\chi_2(\cdot)$, and $\chi'_2(\cdot,\cdot)$ denote circularly symmetric complex Gaussian distribution, beta distribution, chi-squared distribution, and non-central chi-squared distribution, respectively.
  For $0\!\leq\! p\!\leq\!1$, let $h(p) \!=\! -p\ln(p)-(1\!-\!p)\ln(1\!-\!p)$ and $h_2(p) \!=\! h(p)/\ln 2$ with $0\ln 0$ defined to be $0$.
  $\mathbb{N}_{+}$ denotes the set of nonnegative natural numbers.
  For $n\!\in\! \mathbb{N}_{+}$, let $[n] \!=\! \left\{1,2,\cdots\!,n\right\}$.
  Denote the projection matrix on to the subspace spanned by $S$ and its orthogonal complement as $\mathcal{P}_{S}$ and $\mathcal{P}_{S}^{\bot}$, respectively.
  Let $f(x) \!=\! \mathcal{O}\!\left( g(x)\right)$, $x \!\to\! \infty$ mean $\limsup_{x\to\infty}\!\left|f(x)/g(x) \right|\!<\!\infty$ and $f(x) \!=\! o\!\left( g(x)\right)$, $x \!\to\! \infty$ mean $\lim_{x\to\infty}\!\left|f(x)/g(x) \right|\!=\!0$.

\section{System Model} \label{section2}
  We assume the number of UEs $K$ grows linearly and unboundedly with the blocklength $n$, i.e., $K\!=\!\mu n$, $\mu \!<\! 1$, and $n\!\to\! \infty$ \cite{improved_bound}.
  We denote the number of active UEs as $K_a \!=\! p_a K $, where $p_a$ is  the active probability of each UE. The UE set and active UE set are denoted as $\mathcal{K}$ and $\mathcal{K}_a$, respectively.

  We assume each UE has an individual codebook with $M=2^k$ vectors.
  The codebook of the $j$-th UE is denoted as $\mathcal{C}_{j}=\left\{\mathbf{c}_{1}^{j}, \mathbf{c}_{2}^{j},\ldots, \mathbf{c}_{M}^{j}\right\}$, where $\mathbf{c}_{m}^{j}\in \mathbb{C}^{n}$ for $m\in[M]$.
  We consider quasi-static fading channels. The receiver observes $\mathbf{y}$ given by
  \begin{equation}\label{receive_y1}
    \mathbf{y} = \sum_{j\in{\mathcal{K}_a}}{h}_j\mathbf{c}^{j}+\mathbf{z} =\mathbf{AH}\boldsymbol{\beta}+\mathbf{z}\in \mathbb{C}^{n},
  \end{equation}
  where ${h}_j \!\stackrel{i.i.d.}{\sim} \!\mathcal{CN}(0,1)$ denotes the fading coefficient between BS and UE $j$,
  $\mathbf{H}$ is a $KM\!\times\! KM$ block diagonal matrix where block $i$ is an $M\times M$ diagonal matrix with diagonal entries be ${h}_i$, and $\mathbf{z}$ includes i.i.d. $\mathcal{CN}(0,1)$ entries.
  We require the power constraint $\left\|\mathbf{c}_m^{j}\right\|_{2}^{2} \!\leq \!n P$ for $m\in [M]$.
  We denote the signal of UE $j$ as $\mathbf{c}^{j} \!=\! \mathbf{c}_{W_j}^{j}$, where $W_{j} \!\in\! [M]$ is chosen uniformly at random.
  The $((j\!-\!1)M\!+\!1)$-th column to the $(jM)$-th column of $\mathbf{A}$ are codewords of UE $j$.
  The block-sparse vector $\boldsymbol{\beta} \!\in\! \left\{\!\boldsymbol{\beta} \!\in\! \{0,1\}^{\!KM}\!\!: \!\left\| \boldsymbol{\beta}\right\|_{0} \!=\! K_a, \sum_{i=(j-1)M+1}^{jM} \!\boldsymbol{\beta}_{i} \!=\! \{0,1\}, \!\forall j\!\in\! \mathcal{K}\!\right\}$.
  %Let $\hat W_j$ be the estimate of $W_j$ at the decoder.
  The decoder aims to find the estimate $\hat W_j$.
  We use the per-user probability of error (PUPE) $\epsilon$ as performance metric \cite{A_perspective_on}
  \begin{equation}\label{PUPE}
    P_{e}\!=\!\mathbb{E}\!\left[\! \frac{1}{K_a} \!\sum_{j\in {\mathcal{K}_a}} \!1 \!\left[W_{j} \!\neq \!\hat{W}_{j}\right] \!\right]
    \leq \epsilon.
  \end{equation}

  The system achieves the spectral efficiency $S\!=\!p_a\mu k$ and  energy-per-bit $\varepsilon \!= \!\frac{nP}{k} \!=\! \frac{P_{tot,a}}{S}$, where ${P_{tot,a}} \!=\! K_aP$ denotes the total power of active UEs.
  For finite $\varepsilon$, we consider finite $P_{tot,a}$, i.e., $P$ decaying as $\mathcal{O}(1/n)$.
  In this case, we have:
  \begin{defi}\label{defi1}
    An $(n,M,\epsilon,\varepsilon,K,K_a)$ random access code for the channel $ P_{\mathbf{y} | \mathbf{c}^{1}, \ldots, \mathbf{c}^{K_a}}\!:\! \prod_{j\in\mathcal{K}_a} \!\mathcal{C}_j \!\to\! \mathcal{Y}$ is a pair of (possibly randomized) maps including encoders $\left\{ f_j\!: \!\left[M \right] \!\to\!\mathcal{C}_j \right\}_{j=1}^{K}$ and the decoder $g\!: \! \mathcal{Y} \!\to\!\binom {[M]}{K_{a}}$ such that power constraint is satisfied and $P_e\leq\epsilon$. Then, we have the following fundamental limit
    \begin{equation}\label{energy_per_bit_limit}
      \varepsilon^{*}\!(M, \mu, p_a, \epsilon)\!=\!\!\lim_{n \rightarrow \infty} \!\inf \{\varepsilon\!:\! \exists(n, M, \epsilon, \varepsilon, K, K_a)\!-\!\text{ \!\!code }\!\!\},
    \end{equation}
    where the infimum is taken over all possible encoders and decoders, and the limit is understood as $\liminf$ or $\limsup$ depending on whether an upper or a lower bound is given.
  \end{defi}

\section{ Achievability Bound } \label{section3}
\subsection{CSIR} \label{section3_sub1}
  Assuming decoder knows the realization of channel fading coefficients, we can use euclidean metric to decode and obtain:% and obtain Proposition \ref{prop_achi_CSIR}.
  \begin{prop}\label{prop_achi_CSIR}
    Fix spectral efficiency $S$ and target PUPE $\epsilon$.
    Given $\nu\!\in\!(1\;\!\!-\;\!\!\epsilon,\!1]$ and $\epsilon' \!\!=\! \epsilon - 1+\nu$ with CSIR,
    if $\varepsilon \!>\! \varepsilon_{CSIR}^{*}  \!=\! \sup_{\theta\in(\epsilon'\!,\nu]} \! \sup_{\psi\in[0,\nu-\theta]} \! \frac{P_{tot,a}^{'}\!(\theta,\psi)}{S}$, there exists a sequence of $(n, M\!, \epsilon_n, \varepsilon, K\!, K_a)$ codes such that
     $\limsup_{n\to\infty}\epsilon_n\!\leq\!\epsilon$, where
    \begin{equation}\label{P_tot_achi_CSIR}
      P_{tot,a}^{'}\!\!\left(\theta,\!\psi \right) \!=\!\! \frac{4\left({ \exp\!\left\{\!\gamma_{\theta}\!\right\} \!-\! 1} \right)}
      {\xi\!\left( \psi , \!\psi\!+\!\theta \right) \!-\! 4\!\left({ \exp\!\left\{\!\gamma_{\theta}\!\right\} \!-\! 1} \right) \!\xi\!\left( \psi\!+\!\theta , \!\psi\!+\!\theta\!+\!\!1\!\!-\!\nu \right) }\!,
    \end{equation}
    \begin{equation}\label{xi_achi_CSIR}
      {\xi\!\left( \varsigma,\zeta \right)} = \varsigma \ln(\varsigma) - \zeta \ln(\zeta) + \zeta - \varsigma,
    \end{equation}
    \begin{equation}
      \gamma_{\theta} \;\!\!\!=\;\!\!\!p_a \;\!\!\mu h\;\!\!(\;\!\!1\!-\nu+\theta) \!+\;\!\!\mu\!\left(\;\!\!1\!\!-\!\!\nu p_a\!\!+\;\!\!\!\theta p_a \;\!\!\right) \!h\!\;\!\!\left(\;\!\!\! \frac{\theta{p_a}}{{1}\!\!+\!\!\theta{p_a}\!\!\;\!\!-\!\!\nu{p_a}}\;\!\!\!\right)\!+\theta p_a\;\!\!\mu\!\ln\! M \!.
    \end{equation}
  Hence $\varepsilon^{*} (M, \mu, p_a, \epsilon) \leq \varepsilon_{CSIR}^{*}$.
  \end{prop}

  \begin{proof}\label{proof_achi_CSIR}
    We use a random coding scheme by generating a Gaussian codebook for each UE with $\mathbf{c}_{m}^{j} \!\stackrel{\mathrm{i.i.d.}}{\sim} \! \mathcal{CN}\!\left(0,\!P'\mathbf{I}_{n}\right)$ and $P' \!<\! P$.
    The $j$-th UE sends $\mathbf{c}^{j}\!=\!\mathbf{c}_{W_{j}}^{j} \!1\!\left\{\!\left\|\mathbf{c}_{W_{j}}^{j}\right\|_{2}^{2} \!\leq\! n P\!\right\}$ if it is active.
    The decoder estimates the messages of $K_{a_1} \!= \!\nu K_a$ active UEs. Let $?$ denote an error symbol.
    The decoder outputs
      \begin{equation}
        \left[ \hat{\mathcal{K}}_{a_1}, \!\left( \hat{\mathbf{c}}^j \right)_{\!j\in \hat{\mathcal{K}}_{a_1} }\!\right] \!=\! \arg\!\!\!\!\!\!\min_{\stackrel{ \hat{\mathcal{K}}_{a_1} \subset \mathcal{K}} {\!\left| \hat{\mathcal{K}}_{a_1}\!\right| = K_{a_1}} }
        \!\!\!\!\min_{\left( \hat{\mathbf{c}}^j \in \mathcal{C}_j \right)_{j\in\hat{\mathcal{K}}_{a_1}}}
        \!\!\left\|\mathbf{y}- \!\!\!\!\sum_{j\in \hat{\mathcal{K}}_{a_1}}\!\!\!\!h_j \hat{\mathbf{c}}^{j} \right\|_{2}^{2}\!\!,\!
      \end{equation}
      \begin{equation}
        \hat{W}_j = \left\{
        \begin{array}{ll}
          f_{j}^{-1}\left(\hat{\mathbf{c}}^{j}\right) & j \in \hat{\mathcal{K}}_{a_1}  \\
          ? & j \notin \hat{\mathcal{K}}_{a_1}
        \end{array}\right. .
      \end{equation}

    We change the measure over which $\mathbb{E}$ in \eqref{PUPE} is taken to the one with $\mathbf{c}^{j}\!=\!\mathbf{c}_{W_{j}}^{j}$ at the cost of adding $ p_0 \!=\! K_a \mathbb{P}\!\left[\!\frac{\chi_{2}(2n)}{2 n}\!>\!\frac{P}{{P}^{\prime}}\!\right]$ \cite{A_perspective_on}.
    We have $ p_0 \!\to\! 0 \text { as } n \!\to\! \infty$.
    The averaged PUPE becomes
    \begin{equation}\label{PUPE_p0}
      P_{e} \leq p_0 + \mathbb{E}\!\left[\!\frac{1}{K_a} \!\sum_{j\in {\mathcal{K}_a}} \!1\! \left[ W_{j} \!\neq\! \hat{W}_{j} \right] \right]_{\text{\!new measure}} = p_0 + p_1.
    \end{equation}
    Next, we adopt the new measure and omit the subscript.
    Let $F_t\! =\!\! \left\{ \sum_{j\in {\mathcal{K}_a}} \!1\! \left[ W_{j}\! \neq \!\hat{W}_{j} \right] \!=\!K_{a,t} \!\right\}$, $K_{a,t} \!=\! {K_a \!-\!K_{a_1}\!+\!t}$,
    and $\mathcal{T} \!=\! (\epsilon' K_a, \nu K_a]\cap \mathbb{N}_{+}$.
    We can bound $p_1$ as
    \begin{equation}\label{PUPE_p1}
      p_{1} \leq \epsilon +  \mathbb{P} \!\left[ \bigcup_{t \in \mathcal{T}} \!F_t \right]
      \leq \epsilon +  \min \!\left\{\!1, \sum_{t\in \mathcal{T}} \mathbb{P} \!\left[ F_t\right]\!\right\}
      = \epsilon +  p_2.
    \end{equation}
    For simplicity, we rewrite ``$\bigcup_{{S_{1} \subset \mathcal{K}_a,} {\left| S_{1} \right| = K_{a,t}}} \!$'' to ``$\bigcup_{S_{1}}$'' and ``$\bigcup_{{S_{2} \subset \mathcal{K} \backslash \mathcal{K}_a\cup S_{1},} {\left| S_{2} \right| = t}}$'' to ``$\bigcup_{S_{2}}$''; similarly for $\sum$.
    We have
    \begin{align}\label{PUPE_p1_p2_ft}
      \mathbb{P} \!\left[ F_t | \mathbf{H}, \!\mathbf{c}_{[\mathcal{K}_a]}, \mathbf{z} \right]
      \! & \leq\! \mathbb{P}\! \!\left[
      \!\bigcup_{S_{1}} \bigcup_{S_{2}}
      \!\!\!\!\!\bigcup_{\stackrel{ \mathbf{c}^{i'} \!\in \mathcal{C}_i: }{ i \in S_{2}, \mathbf{c}^{i'} \!\!\neq \mathbf{c}^i }}
      \!\!\!\!\!\left\{ \!\left\| \mathbf{z} \!+ \!\! \sum_{i\in S_{1}}\!\!h_i \mathbf{c}^i \!\!- \!\!\sum_{i\in S_{2}}\!\!h_i \mathbf{c}^{i'} \!\right\|_2^2 \right.\right.
      \notag \\
      &\left.\left.\left. \;\;\;\;\;\;\;\leq\!
      \min_{\stackrel{S_{3} \subset S_{1}} {\left| S_{3} \right| = t}} \left\| \mathbf{z} \!+\!\! \!\! \sum_{i\in S_{1} \!\backslash  S_{3} }\!\!\! h_i \mathbf{c}^i \right\|_2^2 \!\right\}  \right| \!\mathbf{H},\! \mathbf{c}_{[\mathcal{K}_a]}\! , \mathbf{z} \right]\notag \\
      & \leq \!\sum_{S_{1}} \sum_{S_{2}}
      \!M^t \mathbb{P} \!\left[ F\!\left( S_{1}, S_{2}, S_{3}^{*}\right) \!| \mathbf{H}, \!\mathbf{c}_{[\mathcal{K}_a]}, \mathbf{z} \right]\!,
    \end{align}
    where $F\!\left( \;\!\!S_{1}\;\!\!, \;\!\!\;\!\!S_{2}\;\!\!, \;\!\!\;\!\!S_{3}^{*}\right) \!\!=\!\! \left\{\!\left\| \mathbf{z}_1 \!\!+\!\! \sum_{i\in \;\!\!S_{3}^{*}}\!\!h_i \mathbf{c}^i \;\!\!\!\;\!\!-\!\!\;\!\! \sum_{i\in \;\!\!S_{2}}\!\!h_i \mathbf{c}^{i'} \!\right\|_2^2 \!\!\;\!\!\leq\;\!\!\!\;\!\!
    \left\| \mathbf{z}_1 \!\right\|_2^2 \!\right\}$ with
    $\mathbf{z}_1 \!=\! \mathbf{z} \!+\! \sum_{i\in S_{1} \! \backslash  S_{3}^{*} } \!h_i \mathbf{c}^i$,
    $\mathbf{c}_{[\mathcal{K}_a]} \!= \!\left\{\mathbf{c}^{i}\!\!: i \!\in\! \mathcal{K}_a \right\}$,
    and $S_{3}^{*}\! \subset \!S_{1}$ is a possibly random subset of size $t$.

    To further bound \eqref{PUPE_p1_p2_ft}, for $\mathbf{a}\sim\mathcal{CN}(0,\mathbf{I}_{n})$, $b\in\mathbb{C}$, $\mathbf{u}\in\mathbb{C}^{n}$, $\gamma>-\frac{1}{\left|b\right|^{2}}$, and $\phi={1+\gamma\left|b\right|^{2}}$, we utilize the identity   \cite{improved_bound}
    \begin{equation}\label{identity_exp}
      \mathbb{E}\left[\exp\left\{-\gamma\|b\mathbf{a}+\mathbf{u}\|_{2}^{2}\right\}\right]
      ={\phi^{-n}}{\exp\left\{-\frac{\gamma} {\phi} {\left\|\mathbf{u}\right\|_{2}^{2}}\right\}}.
    \end{equation}
    The Chernoff bound is also utilized for any random variable $U$, i.e., $\mathbb{P}\left( U\geq d\right) \!\leq\! \min_{\lambda\geq 0} \exp\left\{-\lambda d \right\} \mathbb{E}\left[ \exp\left\{ \lambda U\right\} \right]$ \cite{elements_IT}.
    Hence, let $\lambda_2 = 1+\lambda_1 P' \sum_{i\in S_{2}} \left|h_i  \right|^{2}$, and we have
    \begin{align}\label{PUPE_p1_p2_ft_A}
      &\;\;\;\;\mathbb{P} \!\left[ F\!\left( S_{1}, S_{2}, S_{3}^{*}\right) \!| \mathbf{H}, \mathbf{c}_{[\mathcal{K}_a]}, \mathbf{z} \right] \notag\\
      &\leq\!  \min_{\lambda_1\geq0}
      \left(\lambda_2\right)^{\!-n}
      \!\exp\!\left\{\!\lambda_1 \!\left\|\mathbf{z}_1 \!\right\|_{2}^{2}
      \!-\!\!\frac{\lambda_1}{\lambda_2}  \!\left\| \mathbf{z}_1 \!+\! \sum\nolimits_{i\in S_{3}^{*}}\!\!h_i \mathbf{c}^i \right\|_{2}^{2} \!\right\}\!.\!
    \end{align}
    Taking expectation over $\mathbf{c}_{\left[S_{3}^{*}\right]}$ and $\mathbf{z}_1$, respectively, we have
    \begin{equation}
      \mathbb{P} \!\left[ F\!\left( S_{1},S_{2}, S_{3}^{*}\right) \!| \mathbf{H} \right]
      \!= \!\!\left( \!1\!+\! \frac{P'\!\left( \sum_{i\in S_{2}}\!\! \left|h_i  \right|^{2} \!\!+\!\! \sum_{i\in S_{3}^{*}} \!\! \left|h_i \right|^{2}\!\right)}{4 \!\left( 1\!+\!P' \! \sum_{i\in S_{1}\backslash S_{3}^{*}}\! \left|h_i  \right|^{2} \right) }\! \right)^{\!\!\!\!-n}\!.
    \end{equation}
    We sort $\left\{h_i\!:\!i\!\in\!\mathcal{K}_a\right\}$ in decreasing order of fading power as $\left|{h}_{1}^{\downarrow}\right|\!\! \geq\!\! \left|{h}_{2}^{\downarrow}\right|\!\! \geq \!\!\ldots \!\!\geq \!\! \left|{h}_{K_a}^{\downarrow}\!\right|$.
    Let $ \Psi_{n} \!\!=\! \left[0, \nu\!-\!\theta\right]\!\cap\! \left\{\! \frac{i}{K_a}\!\!: \!i\!\in\! [{K}_a] \!\right\}$.
    Choosing $S_{3}^{*} \!\subset\! S_{1}$ to contain indices with top $t$ fading power, we can obtain
    \begin{align}\label{PUPE_p1_p2_ft2}
      &\mathbb{P} \!\left[ F_t  | \mathbf{H}\right]  \!\leq   \!{\binom {K_a} {K_{a,t}}} \!{\binom {K \!-\!K_{a_1}\!+\!t} {t}}
      M^t \notag \\
      & \;\;\;\;\;\;\;\;\;\;\cdot\! \!
      \left( \!\! \min_{\psi \in \Psi_n}
      \!\!\left\{\!\!1\!+\! \frac{ P' \! \sum_{i=\psi K_a+1}^{\psi K_a+t} \! \left|h_i^{\downarrow}  \right|^{2}}{\!4 \!\left(\!\! 1\!+\!\!P' \! \sum_{i=\psi K_a\!+t+1}^{\psi K_a\!+t+\!K_a\!-\!K_{a_1}}\!\! \left|h_i^{\downarrow}  \!\right|^{2} \right) }\!\! \right\}\!\right)^{\!\!\!\!-n} \!\!.\!
    \end{align}

    Let $\Theta_n \!\!=\!\! (\epsilon'\!, \nu]\!\cap\! \left\{\! \frac{i}{K_a}\!\!: \!i\!\in\! [{K}_a] \!\right\}$ and $t \!=\! \theta K_a$.
    When $\theta \!=\! \nu$, ${\binom {K_a} {K_{a\!,t}}} \!=\! 1$.
    For $\theta \!\in\! \Theta_n\!\backslash \{\nu\}$, we have \cite{Gallager}
    \begin{equation}\label{C1}
      {\binom {K_a} {K_{a,t}}} \!\!\leq\!\! \sqrt{\!\frac{1}{\!2\pi \!K_a(1\!-\!\nu\!+\!\theta) (\nu\!-\!\theta)}} \!\exp\!\left\{\!K_a h(1\!-\!\nu\!+\!\theta)\!\right\}\!.
    \end{equation}
    Let $K_t =  K -K_{a_1}+t$. Similarly, for $\theta \in \Theta_n $, we have
    \begin{equation}\label{C2}
      {\binom {K_t} {t}}
      \!\!\leq\! \sqrt{\!\frac{{1} +\theta{p_a} - \nu{p_a}}{2\pi\theta K_{a}\!\left( {1}\!-\! \nu{p_a} \right)}}
      \!\exp\!\left\{\!K_t h\!\left(\!\frac{\theta{p_a}}{{1}\! +\!\theta{p_a}\!\!-\! \nu{p_a}}\!\right)\!\right\}\!.
    \end{equation}
    For $\tau\!>\!0$ and $0\!<\!\varsigma,\!\zeta \!<\!\!1$, with probability $1\!-\exp\!\left\{\!-\mathcal{O}(n^{\tau}) \right\}$, we have
    $\frac{1}{K_a}\sum_{j=\lceil \varsigma K_a\rceil}^{\lceil \zeta K_a\rceil}\! \left|{{h}^{\downarrow}_{j}} \right|^{2}
      =  \xi \!\left( \varsigma, \zeta \right) + o(1)$ \cite{finite_payloads_fading}.
    We define the event $L_{n}$ with $\mathbb{P}\left[ L_{n}^{c}\right]$ exponentially small in $n$
    \begin{align}\label{event_Ln}
      L_{n}\!\! =\! &\!\bigcap_{\psi \in \Psi_{\!n}}\!\! \left\{\! \!\left\{ \!\!\frac{1}{K_a}\!\!\! \sum_{j= (\psi+\theta) K_a+1}^{ (\!\psi+\theta+\!1\!-\!\nu)\! K_a}\!\! \!\!\!\left|{{h}^{\downarrow}_{j}} \!\right|^{2}
      \!\!\!\!= \! \xi \!\left( \psi\!+\!\theta , \psi\!+\!\theta\!+\!\!1\!\!-\!\nu \right) \!+\! o(1)\!  \!\right\}\!\right. \notag\\
      & \;\;\;\;\;\;\;\;\;\;\bigcap \!\left.\left\{ \!\! \frac{1}{K_a}\!\! \sum_{j= \psi K_a\!+1}^{ (\psi+\theta) K_a}\!\! \left|{{h}^{\downarrow}_{j}} \right|^{2}
      \!\!=\!  \xi \!\left( \psi , \psi\!+\!\theta \right) \!+\! o(1)\!\!\right\}\! \!\right\}\!.\!\!
    \end{align}

    Let $\kappa = 1-\nu p_a+\theta p_a$. We can bound $p_2$ as
    \begin{align}\label{PUPE_p1_p2lim}
      p_2\!&\leq \mathbb{E}\!\left[\min \!\left\{1,\sum_{t\in\mathcal{T}}\mathbb{P} \!\left[ F_t  | \mathbf{H}\right] \right\} 1\left[ L_n\right] \right] + \mathbb{P}\left[ L_n^{c}\right]\notag \\
      & \leq \min \!\left\{\;\!\! 1, \!\sum_{\theta \;\!\!\in\;\!\! \Theta_{\;\!\!n}\!}
      \!\exp\!\left\{\!o(n) \!-\!n \!
      \left(\!-\mu\kappa h\!\left(\!\frac{\theta{p_a}}{\kappa}\!\right)
      \!-\! \theta p_a \mu \ln \!M
       \right. \right. \right.\notag \\
      & \;\;\;\;\;\;\;\;\;\;\;\;\;+ \!\min_{\psi \in\;\!\! \Psi_{\;\!\!n}}\!\ln \!\left(
      \!1\!+\! \frac{ P'\!K_a \xi \!\left( \psi , \psi\!+\!\theta \right) }{4 \left( 1\!+\!P'\!K_a  \xi \left( \psi\!+\!\theta , \;\!\!\psi\!+\!1\!-\!\nu\!+\!\theta \right) \right) }\right) \notag \\
      &  \;\;\;\;\;\;\;\;\;\;\;\;\left.\left.\left. - p_a \mu h\!\left( 1-\nu+\theta \right) \right)\right\} \right\} + o(1).
    \end{align}
    Define $\Theta =  (\epsilon', \nu]$ and $\Psi =  [0, \nu-\theta]$. Choosing $K_aP' > \sup_{\theta\in \Theta} \! \sup_{\psi\in \Psi} \! P_{tot,a}^{'}\!\left(\theta,\!\psi\right)$ will ensure $\limsup_{n\to\infty}\! p_2 \!=\! 0$.
  \end{proof}

\subsection{No-CSI} \label{section3_sub2}
  In this section, we assume neither the transmitters nor the decoder knows the realization of fading coefficients, but they both know the fading distribution. In this case, we can obtain:
  \begin{prop}\label{prop_achi_noCSI}
    Fix spectral efficiency $S$ and target PUPE $\epsilon$.
    With no-CSI, if $\varepsilon > \varepsilon_{no-CSI}^{*}  = \sup_{\theta\in(\epsilon,1]}  \frac{P_{tot,a}^{'}(\theta)}{S}$,
    there exists a sequence of $(n,\! M\!, \epsilon_n, \varepsilon, K, \!K_a)$ codes such that $\limsup_{n\to\infty}\!\epsilon_n\!\leq\!\epsilon$, where
    \begin{equation}\label{P_tot_achi_noCSI}
      P_{tot,a}^{'}\left(\theta \right) = \frac{W_{\theta}}{\left(1-\delta_{3}^{*}\right)\xi\left( 1-\theta,1\right)},
    \end{equation}
    \begin{equation}\label{W_achi_noCSI}
      W_{\theta} = \frac{1\!-\!V_{\theta}}{V_{\theta}}(1+\delta_{2,\theta}) ,
    \end{equation}
    \begin{align}\label{V_achi_noCSI}
      {V}_{\theta} \!&=\! \exp \!\left\{ \!-\! \left( \!\delta_{1,\theta}^*
      \!+\! \frac{1-p_a \mu +\theta p_a \mu}{1- p_a \mu}h\left( \frac{\theta p_a \mu}{1-p_a \mu +\theta p_a \mu}\right) \right.\right. \notag\\
      & \;\;\left.\left. \!+ \frac{ \theta p_a \mu \ln M}{1-p_a\mu}
      \!+\! \frac{\mu(1\!-\!p_a \!+\!\theta p_a )}{1- p_a \mu} h\!\left( \frac{\theta p_a}{\!1 \!-\! p_a \!+\!\theta p_a\!} \right) \!\right)\!\right\}\! ,
    \end{align}
    \begin{equation}\label{delta_achi_noCSI}
      \delta_{1,\theta}^* = \frac{p_a \mu}{1-p_a \mu} h(\theta),
    \end{equation}
    \begin{equation}\label{c_achi_noCSI}
      c_\theta = \frac{2 V_{\theta}}{1-V_{\theta}} ,
    \end{equation}
    \begin{equation}\label{q_achi_noCSI}
      q_\theta = \frac{ p_a \mu }{1-p_a \mu+ \theta p_a \mu} h(\theta),
    \end{equation}
    \begin{equation}\label{delta1_achi_noCSI}
      \delta_{2,\theta}^* = q_\theta\left( 1+c_\theta\right) + \sqrt{q^2_\theta c_\theta \left( 2+c_\theta\right) + 2q_\theta\left( 1+c_\theta\right)},
    \end{equation}
    \begin{equation}\label{delta2_achi_noCSI}
      \delta_{3}^* \!=\! \inf \!\left\{ x\!: 0<x<1, -\ln(1-x)-x > 0\right\},
    \end{equation}
  where $\xi(\cdot,\cdot)$ is given in \eqref{xi_achi_CSIR}.
  Hence, $\varepsilon^{*} (M, \mu, p_a, \epsilon) \!\leq\! \varepsilon_{no-CSI}^{*}$.
  \end{prop}

  \begin{proof}\label{proof_achi_noCSI}
    We assume each UE has a Gaussian codebook with power $P' \!<\! P$. The $j$-th UE transmits $\mathbf{c}^{j}$ if it is active.
    We utilize the projection decoder based on the fact that $\mathbf{y}$ in \eqref{receive_y1} belongs to the subspace spanned by the transmitted signals if the additive noise is neglected \cite{Beta_dis}.
    The output is given by
    %$\mathbf{c}^{j}=\mathbf{c}_{W_{j}}^{j} 1\left\{\left\|\mathbf{c}_{W_{j}}^{j}\right\|_{2}^{2} \leq n P\right\}$
      \begin{equation}\label{decoder_achi_noCSI}
        \left[ \hat{\mathcal{K}}_a, \!\left( \hat{\mathbf{c}}^j \!\;\!\right)_{\!j\!\;\!\in\!\;\! \hat{\mathcal{K}}_a }\!\right] \!\!=\! \arg\!\!\!\!\!\!\!\!\max_{{ \hat{\mathcal{K}}_a \subset \mathcal{K}}, {\left| \hat{\mathcal{K}}_a\!\right| = K_a} } \max_{\left( \hat{\mathbf{c}}^j \!\;\!\in \!\;\!\mathcal{C}_j \!\right)_{\!j\!\;\!\in\!\;\!\hat{\mathcal{K}}_{\!\;\!a}}} \!\!\!\left\| \mathcal{P}_{\!\left\{\!\hat{\mathbf{c}}^j\!: j\in \hat{\mathcal{K}}_a \!\right\}} \!\mathbf{y}\right\|_2^2\!\!,\!\!
      \end{equation}
      \begin{equation}\label{decoder_W_achi_noCSI}
        \hat{W}_j = \left\{
        \begin{array}{ll}
          f_{j}^{-1}\left(\hat{\mathbf{c}}^{j}\right) & j \in \hat{\mathcal{K}}_a  \\
          ? & j \notin \hat{\mathcal{K}}_a
        \end{array}\right. .
      \end{equation}
    
    As in Section \ref{section3_sub1}, we change the measure and obtain \eqref{PUPE_p0}. Next, we bound $p_1$ as in \eqref{PUPE_p1} with $p_2 \!=\!\min\left\{1, \!\sum_{t\in \mathcal{T}} \!\mathbb{P} \left[ F_t\right]\right\}$,
    $F_t\! =\! \left\{\! \sum_{j\in {\mathcal{K}_a}} \!\!1\! \left[ W_{j}\! \neq \!\hat{W}_{j} \right] \!=\! t \!\right\}$, and $\mathcal{T} \!\!=\!\! (\epsilon K_a, K_a]\cap \mathbb{N}_{+}$.
    Let $F\!\left( S_{1}, S_{2}\right) \!=\! \left\{ \left\| \mathcal{P}_{\!\mathbf{c}_{\left[S_{2}\right]}^{'}, \mathbf{c}_{ \left[\mathcal{K}_a \!\backslash S_{1}\right] } } \mathbf{y} \right\|_2^2 \!\geq\! \left\| \mathcal{P}_{\!\mathbf{c}_{ \left[\mathcal{K}_a \right] } } \mathbf{y} \right\|_2^2 \right\}$.
    We rewrite ``$\bigcup_{{S_{1} \subset \mathcal{K}_a},{\left| S_{1} \right| = t}}$'' to ``$\bigcup_{S_{1}}$'' and ``$\bigcup_{{S_{2} \subset \mathcal{K} \backslash \mathcal{K}_a\cup S_{1}}, {\left| S_{2} \right| = t}}$'' to ``$\bigcup_{S_{2}}$''; similarly for $\sum$.
    We have
    \begin{equation}\label{PUPE_p1_p2_ft_noCSI}
      \mathbb{P} \!\left[ F_t | \mathbf{H}, \!\mathbf{c}_{[\mathcal{K}_a]}, \mathbf{z} \right]
      \! \!=\! \mathbb{P}\! \!\left[\! \left.
      \bigcup_{S_{1}}
      \bigcup_{S_{2}}
      \!\!\!\!\!\!\bigcup_{\stackrel{ \mathbf{c}^{i'} \!\in \mathcal{C}_i: }{ i \in S_{2}, \mathbf{c}^{i'} \!\neq \mathbf{c}^i }}
      \!\!\!\!\!\!\!\!\!\! F\!\left( S_{1}, S_{2}\right)  \right| \!\mathbf{H}, \mathbf{c}_{[\mathcal{K}_a]}, \mathbf{z} \right]\!\!.\!
    \end{equation}

    Let $A_1 \!=\! \mathbf{c}_{ \left[\mathcal{K}_a \!\backslash S_{1}\right] }$, $B_1 \!=\! \mathbf{c}_{\left[S_{2}\right]}^{'}$, and $V \!=\! \operatorname{span}\{A_1,B_1\} \!=\! A \oplus B$ where $A$ and $B$ are subspaces of dimension $K_a\!- t$ and $t$ respectively, with $A\!=\! \operatorname{span}(A_1)$ and $B$ is orthogonal complement of $A_1$ in $V$.
    Hence, $\left\|\mathcal{P}_{V} \mathbf{y}\right\|_2^2 \!=\! \left\|\mathcal{P}_{\!A_1} \mathbf{y}\right\|_2^2  +  \left\|\mathcal{P}_{\!B}\mathcal{P}_{\!A_1}^{\bot} \mathbf{y}\right\|_2^2 $.
    Denote $n_t \!= n\!-\!K_a\!+t$. 
    Conditioned on $\mathbf{H}$, $\mathbf{c}_{[\mathcal{K}_a]}$, and $\mathbf{z}$, the law of
    $\left\|\mathcal{P}_{\!B} \mathcal{P}_{\!A_1}^{\bot} \mathbf{y}\right\|_2^2$ is the law of squared length orthogonal projection of a fixed vector in $\mathbb{C}^{n_t}$ of length $\left\|\mathcal{P}_{\!A_1}^{\bot} \mathbf{y}\right\|_2^2$ onto a (uniformly) random $t$ dimensional subspace.
    It is the same as the law of squared length orthogonal projection of a random vector of length $\left\|\mathcal{P}_{\!A_1}^{\bot} \mathbf{y}\right\|_2^2$ in $\mathbb{C}^{n_t}$ onto a fixed $t$ dimensional subspace, i.e., $\left\|\mathcal{P}_{\!A_1}^{\bot} \mathbf{y}\right\|_2^2 \beta(t,n\!-\!K_a)$ \cite{Beta_dis}.
    We have
    \begin{equation}\label{PUPE_p1_p2_fst_noCSI}
      \mathbb{P} \left[ \left.
      F\left( S_{1}, S_{2}\right)  \right| \mathbf{H}, \mathbf{c}_{[\mathcal{K}_a]}, \mathbf{z} \right]
      =F_{\beta}(G_{S_{1}}; n-K_a, t),
    \end{equation}
    where $ G_{S_{1}} \!\!=\! { \left\|\mathcal{P}_{\mathbf{c}_{[\mathcal{K}_a]}}^{\bot}\mathbf{z} \right\|_2^2 }\Big/
      {\left\|\mathcal{P}_{\mathbf{c}_{[\mathcal{K}_a \!\backslash S_{1}]}}^{\bot}\mathbf{y} \right\|_2^2 }$.
    $F_{\beta}(G_{S_{1}}; n-K_a, t)$ denotes the CDF of beta distribution with parameters $n-K_a$ and $t$ satisfying $F_{\beta}(G_{S_{1}}; n-K_a, t)\!\leq\! {\binom {n_t-1} {t-1}} G_{S_{1}}^{n-\!K_a}$ for $t \!\geq \!1$.

    Denote $K_t = K-K_a+t$. Then, \eqref{PUPE_p1_p2_ft_noCSI} can be bounded as
    \begin{equation}\label{PUPE_p1_p2_ftb_noCSI}
      \mathbb{P} \!\left[ F_t | \mathbf{H},  \mathbf{c}_{[\mathcal{K}_{a} ]}, \mathbf{z} \right]
      \! \leq \!\min \!\left\{\!1, \sum_{S_{1}} {\!\binom {K_t} {t}}  M^t \!{\binom {n_t - 1} {t -1}} G_{S_{1}}^{n-K_a}\!\right\} \!.
    \end{equation}
    Let $t \!=\! \theta K_a$ and $ \Theta_n \!=\! (\epsilon, 1]\!\cap \!\left\{ \!\frac{i}{K_a}\!: i\!\in\! [{K}_a] \right\}$.
    We can bound ${\binom {K_a} {t}} $ and ${\binom {K_t} {t}}$ based on \cite{Gallager}. Meanwhile, we have
    \begin{equation}
      {\binom {n_t-1} {t-1}}\!\leq\!
      \sqrt{\frac{n_t-1}{2\pi(t-1) \left( n-\!K_a \right)}} \exp \!\left\{n_t \;\!h\!\left({\frac{t}{n_t}}\right)\right\}.
    \end{equation}

    We denote $r_{\theta} \! =\! \frac{ \theta p_a \mu \ln M}{1- p_a\mu} 
    \!+\! \frac{1-p_a \mu +\theta p_a \mu}{1- p_a \mu}h\left( \frac{\theta p_a \mu}{1-p_a \mu +\theta p_a \mu}\right)
    \!+ \! \frac{\mu-p_a \mu +\theta p_a \mu}{1- p_a \mu} h\!\left( \frac{\theta p_a}{1 - p_a +\theta p_a} \right)    
    \!+\! \frac{\ln \left({\frac{1 +\theta p_a - p_a}{2\pi\theta K_a\!\left( 1 - p_a \right)} } \right)}{2(n-K_a)}
    \!+\! \frac{\ln \left( \frac{n_t-1}{2\pi(t-1)\!\left( n-K_a \right)} \right)}{2(n-K_a)} $,
    $\tilde{V}_{n,{\theta}} \!=\! r_{\theta} + \delta_{1,\theta}$ with $\delta_{1,\theta}\!>\!0$, and ${V}_{n,{\theta}} \!=\! \exp\!\left\{ \!-\tilde{V}_{n,{\theta}}\!\right\}$.
    We have $\lim_{n\to\infty} \!{V}_{n,\theta} \!=\! {V}_{\theta} $ as in \eqref{V_achi_noCSI}.
    Define the event
    $ L_{1} \!=\! \bigcap_{t \in \mathcal{T}}  \bigcap_{S_{1}} \left\{G_{S_{1}}\!\leq\! V_{n, \theta}\right\}$.
    Then, $p_{2}$ can be bounded as
    \begin{align}\label{PUPE_p1_p2b_noCSI}
      p_{2} \!&\leq\! \mathbb{E} \! \left[ \!\min\!\left\{\!1, \!\sum_{t \in \mathcal{T}} \sum_{S_{1}}
      \!\exp\!\left\{ (n\!-\!\!K_a) r_{\theta} \right\} \!G_{S_{1}}^{n\!-\!K_a} \! \!\right\}\! 1\!\left[ L_{1} \right] \right]\!\!+\! \mathbb{P}\!\left[ L_{1}^c \right] \notag \\
      & \leq \!\!\sum_{\theta \in \Theta_n \backslash \{1\}} \!\!\!\!\!\! \exp\!\left\{\!K_a h(\theta) \!-\! (n\!-\!\!K_a)\delta_{1,\theta} \!-\! \frac{1}{2}\!\ln\! \left(2\pi\theta K_a(1\!-\!\theta) \right)\!\right\}\notag \\
      &\;\;\;\;+ \exp\!\left\{- n(1\!-\!p_a \mu)\delta_{1,\theta=1}\right\}
      + \mathbb{P}\left[ L_{1}^c \right].
    \end{align}

    Let $p_{3} = \mathbb{P}\left[ L_{1}^c \right]$, which can be bounded as
    \begin{align}\label{PUPE_p1_p2_p3a_noCSI}
      p_{3}\!&\leq \!\mathbb{P} \!\left[\! \bigcup_{t\in \mathcal{T}} \bigcup_{S_{1}}
      {\left\|\mathcal{P}^{\bot}_{\!\mathbf{c}_{[\mathcal{K}_a \!\backslash\! S_{1}]}}\!\mathbf{z} \right\|_2^2 }
      \!\!>\! V_{n, t} {\left\| \mathcal{P}^{\bot}_{\!\mathbf{c}_{[\mathcal{K}_a \!\backslash \!S_{1}]}} \!\!\left( \sum_{i\in S_{{1}}}\!\!{h_i \mathbf{c}_i  } \!+\! \mathbf{z} \! \right)\!\right\|_2^2 } \right] \notag \\
      & =\! \mathbb{P} \!\!\left[\!\bigcup_{t \in \mathcal{T}} \! \bigcup_{S_{1}}
      \!\left\|Q_{S_{1}}\!\right\|_2^2
      \!>\!\! \frac{V_{n, t}}{\left( 1\!\!-\!\!V_{n,t} \right)^2} {\left\| \mathcal{P}^{\bot}_{\!\mathbf{c}_{[\mathcal{K}_a \! \backslash \! S_{1}]}} \!\! \sum_{i\in S_{1}} \!\!{h_i \mathbf{c}_i  } \right\|_2^2 } \right]\!.\!
    \end{align}
    where $Q_{S_{1}}\!\!\!= \!\! \mathcal{P}^{\bot}_{\mathbf{c}_{[\mathcal{K}_a \! \backslash \! S_{1}]}} \!\!\left( \!\mathbf{z} \!-\! \frac{V_{n, t}}{1-V_{n,t}} \! \sum_{i\in S_{1}} \!\!{h_i \mathbf{c}_i  } \right)$.
    Conditioned on $\mathbf{H}$ and $\mathbf{c}_{\left[\mathcal{K}_a \right]}$, we have ${\left\|Q_{S_{1}}\!\right\|_2^2 } \!\sim\! \frac{1}{2} \chi'_2 \!\left( 2\lambda, 2n_t\right)$ with conditional expectation $\mu \!=\! \lambda\!+\!n_t$ and    
    $\lambda \!=\! {\left\| \! \frac{V_{n, t}}{1-V_{n,t}} \!\mathcal{P}^{\bot}_{\mathbf{c}_{[\mathcal{K}_a \! \backslash \! S_{1}]}} \!\sum_{i\in S_{1}} \!\!{h_i \mathbf{c}_i  } \right\|_2^2 }$.

    Denote $U \!=\! \frac{V_{n, t}}{1-V_{n,t}} \!\left\| \mathcal{P}^{\bot}_{\mathbf{c}_{[\mathcal{K}_a \! \backslash \! S_{1}]}} \sum_{i\in S_{1}} \!\!{h_i \mathbf{c}_i  } \right\|_2^2 \!-\! n_t\!=\!n_t U^1$ and $T \!= \!\frac{1}{2} \chi'_2 \!\left( 2\lambda, 2n_t\right) - \mu$. %We have $\lambda \!=\! \frac{V_{n, t}}{1-V_{n,t}} n_t (1+U^1)$.
    Define the event $L_{2} \!=\! \bigcap_{t \in \mathcal{T}} \bigcap_{S_{1}} \!\left\{ U^1 \!\geq\! \delta_{2,\theta}\right\}$ with $\delta_{2,\theta}\!>\!0$.
    Then, we can obtain
    \begin{equation}\label{PUPE_p1_p2_p3b_noCSI}
      p_{3} \!\leq \!
      \sum_{t \in \mathcal{T}}  \sum_{S_{1}}  \mathbb{E} \!\left[ \mathbb{P} \!\left[\left. T\! >\! U \right| \mathbf{H},\! \mathbf{c}_{\left[\mathcal{K}_a \right]} \right] \!1\!\left[  U^1 \!\!\geq\! \delta_{2,\theta} \right]  \right]
      \!+ \mathbb{P} \!\left[ L_{2}^c\right]\!.\!
    \end{equation}

    To further bound $p_3$, we use the following concentration results \cite{concentration_ineq1, concentration_ineq2}.
    Let $\chi \sim \chi_2(d)$. Then $\forall x >1 $, we have
    \begin{equation}
      \mathbb{P}\left[\chi \leq \frac{d}{x}\right] \leq \exp\left\{-\frac{d}{2}\left(\ln x+\frac{1}{x}-1\right)\right\}.
    \end{equation}
    Let $\chi \sim \chi'_2(a, d)$. Then $\forall x >0 $, we have
    \begin{equation}
      \mathbb{P}[\chi\!\geq\! x\!+\!a\!+\!d\;\!]
      \!\leq \! \exp\!\left\{\!\!-\frac{1}{2}\!\!\left(\!x\!+\!d\!+\!2 a\!-\!\!\sqrt{\!d\!+\!2 a\!} \sqrt{2 x\!+\!d\!+\!2 a}\right)\!\!\right\}\!.
    \end{equation}
    Hence, we have $\mathbb{P}\!\left[\left. T \!>\! U \right| \mathbf{H}, \!\mathbf{c}_{\left[\mathcal{K}_a \right]} \right]  \!\!\leq \!\! \exp\!\left\{\!-n_t f_{n,\theta}(U^1)\!\right\}$.
    Let $V'_{n,\theta} \!=\! \frac{2 V_{n, \theta}}{1\!-\!V_{n,\theta}}$.
    For $0\!<\!V_{n, \theta}\!<\!1$ and $x\!>\!0$, $f_{n,\theta}(x)$ is given by
    \begin{align}
      f_{n,\theta}(x) \!=& (1+V'_{n,\theta})(1+x)  \notag\\
      & \!-\!\! \sqrt{1\!+\!V'_{n,\theta}(1\!+\!x)} \sqrt{1\!+\!2x\!+\!V'_{n,\theta}(1\!+\!x)}>0.
    \end{align}
    It is a monotonically increasing function of $x$.
    Then, we have
    \begin{align}\label{PUPE_p1_p2_p3c_noCSI}
      p_{3}\!&\leq \sum_{\theta \in \Theta_n\!\backslash\!\{1\}} \!\!\!\!\!\! \exp\!\left\{\!K_a h(\theta) \!- \! n_tf_{n,\theta}(\delta_{2,\theta}) \!-\! \frac{1}{2}\!\ln \!\left(2\pi\theta K_a(1\!-\!\theta) \right)\!\right\}\notag \\
      & \;\;\;\; + \exp\left\{- n_tf_{n,\theta=1}(\delta_{2,\theta=1}) \right\} \!+\! \mathbb{P}\left[ L_{2}^c \right].
    \end{align}

    We have $\left\| \mathcal{P}^{\bot}_{\mathbf{c}_{[\mathcal{K}_a \! \backslash \! S_{1}]}} \!\sum_{i\in S_{1}} \!\!{h_i \mathbf{c}_i  } \right\|_2^2 \!\!\sim\! \frac{P'}{2} \!\sum_{i\in S_{1}}  \!\!{\left|h_i \right|^2}\!\;\! \chi_2\!\left(2n_t\right)$ conditioned on $\mathbf{H}$.
    Let $L_3 \!=\! \!\left\{ \!\frac{\chi_2\left(2n_t\right)}{2n_t} \!\geq\! 1\!-\!\delta_{3}\!\right\}$ with $0\!<\!\delta_{3}\!<\!1$ and we have $p_4 = \mathbb{P}\left[ L_3^c \right] \leq \exp\left\{-n_t \left(- \ln \left( {1-\delta_{3}} \right) - \delta_{3}\right)\right\} $.
    Given $W_{\theta} $ in \eqref{W_achi_noCSI}, we can bound $p_5 = \mathbb{P} \left[ L_{2}^c\right]$ as
    \begin{align}\label{PUPE_p1_p2_p3_p4_noCSI}
      p_{5}\!&\leq \mathbb{P}\!\left[\bigcup_{t \in \mathcal{T}} \bigcup_{S_{1}}
      \left\{ \! P'\!\sum_{i\in S_{1}}  \!{\left|h_i \right|^2} \frac{\chi_2\!\left(2n_t\right)}{2n_t}
      \!<\! W_{\theta} \! \right\}\!\cap\! \left\{ L_3 \right\} \!\right]
      \!+\! \mathbb{P} \left[ L_{3}^c\right] \notag \\
      & \leq  \sum_{t \in \mathcal{T}} \mathbb{P}\!\left[
      P'\!\!\sum_{i=K_a\!-t+1}^{K_a}  \!{\left|h_i^{\downarrow} \right|^2}
      \!<\!\frac{W_{\theta}}{1-\delta_{3}} \!\right]
      + p_4\notag\\
      &= p_6 + p_4.
    \end{align}
    Define the event $L_{4}\!\! =\!\!\left\{ \!\!\frac{1}{K_a}\! \!\sum_{j= K_a\!-t+1}^{ K_a}\!\left|{{h}^{\downarrow}_{j}} \!\right|^{2}
    \!\!\!\!= \! \xi \!\left( 1\!-\!\theta, \!1 \right) \!+\! o(1)\!  \right\}$ with $\mathbb{P}\left[ L_{4}^{c}\right]$ exponentially small in $n$. $p_6$ can be bounded as
    \begin{align}\label{PUPE_p1_p2_p3_p6_noCSI}
      p_{6} &\leq\! \sum_{t \in \mathcal{T}} \mathbb{P}\!\left[\!
      \left\{\!P'\! \!\!\!\sum_{i=K_a\!-t+1}^{K_a}  \!\!{\left|h_i^{\downarrow} \right|_2^2}
      \!<\!\frac{W_{\theta}}{1-\delta_{3}} \!\right\}\cap \left\{ L_4 \right\}\!\right]
      \!+ \!\sum_{t \in \mathcal{T}} \mathbb{P}\!\left[ L_4^c \right]\notag \\
      & \leq\! \sum_{t \in \mathcal{T}} 1\!\left[
      {P' \!K_a \!\left(\xi \left( {1-\theta},1\right)\!+\!o(1)\right) }
      \!<\!\frac{W_{\theta}}{1\!-\!\delta_{3}}\right]
      \!+\! o(1).
    \end{align}

    Hence, $p_{2}$ can be bounded as
    \begin{align}\label{PUPE_p_suma_noCSI}
      p_{2} \leq & \sum_{\theta \in \Theta_n } \left\{ \exp\left\{ o(n) - n \left((1- p_a\mu)\delta_{1,\theta} - p_a\mu h(\theta) \right)\right\}  \right.\notag \\
      &  +  \exp\left\{ o(n) \!-\! n \!\left((1\!-\!p_a\mu\!+\!\theta p_a \mu)f_{n,\theta}(\delta_{2,\theta}) \!-\! p_a\mu h(\theta) \right)\right\} \notag \\
      &  + \left. \!1\!\left[
      {P' \!K_a \!\left(\xi \! \left( 1\!-\!\theta ,1\right)\!+\!o(1)\right)}\!<\! \!\frac{W_{\theta}}{1\!-\!\delta_{3}} \right] \!\right\} \notag \\
      & + \exp\left\{-n_t \left(- \!\ln\! \left( {1\!-\!\delta_{3}} \right) \!-\! \delta_{3} \right)\right\} + o(1).
    \end{align}
    For $\forall \theta \!\in\! \Theta \!=\! (\epsilon,1]$, choosing $\delta_{1,\theta} \!>\! \delta_{1,\theta}^{*}$, $\delta_{2,\theta} \!>\! \delta_{2,\theta}^{*}$, $\delta_{3} \!>\! \delta_{3}^{*}$, and $K_aP' \!>\! P_{tot,a}^{'}\left(\theta\right)$ will ensure $\limsup_{n\to\infty} p_2 \!=\! 0$.
  \end{proof}

\section{ Converse Bound } \label{section4}
\subsection{CSIR} \label{section4_sub1}
  \begin{prop}\label{prop_converse_CSIR}
    We assume spectral efficiency $S$ and target PUPE $\epsilon$ are fixed.
    With CSIR, we can obtain $\varepsilon^{*} (M, \mu, p_a, \epsilon) \geq \inf \!\frac{P_{tot,a}}{S}$, where infimum is taken over all ${P_{tot,a}}\!>\!0$ satisfying
    \begin{equation}\label{P_tot_conv_CSIR}
      P_{tot,a} \!\geq
      \frac{2^{ p_a \mu \theta k - p_a \mu \epsilon \! \log_2(M\!-\!1) - p_a \mu h_2(\epsilon) }\!-\!1 }
      {\xi\!\left(1-\theta,1 \right) }, \forall \theta\!\in\!(0,1],
    \end{equation}
    \begin{equation}\label{P_tot_conv_CSIR_singleUE}
      \epsilon \geq 1-\mathbb{E}\left[Q\left(Q^{-1}\left(\frac{1}{M}\right)-\sqrt{\frac{2 P_{t o t,a}}{p_a\mu}|h|^{2}}\right)\right].
    \end{equation}
  \end{prop}

  \begin{proof}\label{proof_converse_CSIR}
    Let $\mathcal{W}_{\mathcal{K}_a} \!\!=\!\! \left\{W_i\!: i\!\in\! {\mathcal{K}_a} \!\right\}$ be the sent messages of active UEs,
    $\mathcal{X}_{\mathcal{K}_a} \!\!=\!\! \left\{\mathbf{c}^i\!: i\!\in\! {\mathcal{K}_a} \!\right\}$ be corresponding codewords,
    $\mathbf{y}_{\mathcal{K}_a}$ be the received vector, and $\hat{\mathcal{W}}_{\mathcal{K}_a} \!\!=\!\! \left\{\hat{W}_i\!: i\!\in\! {\mathcal{K}_a} \!\right\}$ be the decoded messages.
    We have the Markov chain: $\mathcal{W}_{\mathcal{K}_a} \!\!\!\to\!\! \mathcal{X}_{\mathcal{K}_a} \!\!\!\to\!\! \mathbf{y}_{\mathcal{K}_a} \!\!\!\to\!\! \hat{\mathcal{W}}_{\mathcal{K}_a}$.
    We assume a genie reveals the set $\mathcal{K}_a$ of active UEs and a set $S_1 \!\subset\! \mathcal{K}_a$ for messages $\mathcal{W}_{S_1} \!=\! \left\{W_i\!: i\!\in\! {S_1} \right\}$ and fading coefficients ${h}_{S_1} \!=\! \left\{h_i\!: i\!\in\! {S_1} \right\}$ to the decoder. Let $S_2 \!=\! \mathcal{K}_{a}\!\backslash S_1$ with $\left|S_2\right| \!=\! \theta K_a$ and $\theta \!\in \!\Theta_n \!=\! (0, 1]\cap \!\left\{\! \frac{i}{K_a}\!\!: \!i\!\in\! [{K}_a] \right\} $. The equivalent received message is given by
    \begin{equation}\label{receive_y_G}
      \mathbf{y}^G = \sum_{i\in{S_2}}{h}_i\mathbf{c}^{i}+\mathbf{z} \in \mathbb{C}^{n}.
    \end{equation}
    Denote the decoded message for the $i$-th UE with genie as $\hat{W}_{i}^{G}$.
    Let $L_i \!=\! 1\!\left[ {W}_{i}\!\neq\! \hat{W}_{i}^{G}\right]$ and $P_{e,i}^G \!=\! \mathbb{E}\left[ L_i \right]$. We have $P_{e,i}^G \!=\! 0$ for $i\!\in\! S_1$.
    The averaged PUPE is $P_{e}^G\!=\!\frac{1}{K_a} \!\sum_{i\in {S_2}} \!P_{e,i}^G \leq \epsilon$.

    For $i \in S_2$, based on Fano inequality, we have
    \begin{equation}\label{fano}
      \log_2M \!-\! P_{e,i}^G \log_2(M-1) \!-\! h_2\!\left( P_{e,i}^G \right) \!\leq\! I\left(W_i;\hat{W}_i^{G} \right).
    \end{equation}
    Considering $\sum_{i\in S_2} \!I\!\left(\!W_i; \;\!\!\hat{W}_i^{G} \!\right) \!\leq\! n \mathbb{E}\!\left[\log_2\!\;\!\!\left( \!1\!+\!\;\!\! P \!\sum_{i\in S_2} \!\!\left| h_i\right|^2\right) \!\right]$
    and the concavity of $h_2$, we can obtain
    \begin{equation}
      {\theta}k \!-\! P_{e}^G \!\log_2(M\!-\!1) \!-\! h_2\!\left( P_{e}^G \right)
      \!\leq \!\! \frac{n}{K_a} \mathbb{E}\!\left[\log_2\!\!\left( \!\!1\!+\! P \!\sum_{i\in S_2} \!\left| h_i\right|^2\!\!\right) \!\right] \!\!.
    \end{equation}
    Since $P_{e}^G \leq \epsilon \leq 1-\frac{1}{M}$, we have $P_{e}^G \log_2(M-1) + h_2\left( P_{e}^G \right) \leq \epsilon \log_2(M-1) + h_2\left( \epsilon \right) $.
    We can obtain
    \begin{equation}
      {\theta k}\;\!\! -\;\!\! \epsilon \log_2(M\!\!\;\!-\;\!\!1) \;\!\!-\!\;\! h_2\!\left( \epsilon \right) \!\!\leq\! \! \frac{n}{K_a} \! \mathbb{E}\!\!\left[\!\log_2\!\!\left( \!\!1\!+\! P \!\!\!\!\!\!\!\sum_{i=(1\!-\!\theta) \!K_a}^{K_a} \!\!\!\!\left| h_i^{\downarrow}\!\right|^{2}\!\right)\! \!\right]\!\!.\!
    \end{equation}

    For $a,b\!\in\! (0,1]$, let $S_{K_a}\!(a,b) \!=\! \frac{1}{K_a}\!\!\sum_{i=a{K_a}}^{b{K_a}} \!\!\left| h_i^{\downarrow}\right|^2 $ satisfying $S_{K_a}\!(a,b) {\to} \xi(a,b)$ as ${K_a}\!\!\!\to \!\!\!\infty$ and $\mathbb{E}\!\left[S_{K_a}\!(a,b)\right]\!\!\! \leq\!\!\! 1$.
    The family of random
    variables $\left\{ S_{K_a}\!(a,b)\! :\! {K_a}\!\in \! \mathbb{N}_{+}\!\right\}$ is uniformly integrable
    based on the dominated convergence theorem \cite{uniformly_integrable}.
    Since $0\!\!<\!\!\log_2\left( 1\!+\!P_{tot,a}S_{K_a}\!(a,b)\right) \!\!<\!\! P_{tot,a}S_{K_a}\!(a,b)$, the family $\left\{ \log_2\!\left( 1\!+\!P_{tot,a}S_{K_a}(a,b)\right) \!: \!{K_a}\!\in \! \mathbb{N}_{+}\right\}$ is also uniformly integrable.
    As ${K_a}\!\!\to \!\!\infty$, since $\log_2\!\left( 1\!+\!P_{tot,a}S_{K_a}\!(a,b)\right) \!\!\to\!\! \log_2\!\left( 1\!+\!P_{tot,a} \xi(a,b)\right)$,
    we have $\mathbb{E}\!\left[ \log_2\!\left( 1\!+\!P_{tot,a}S_{K_a}\!(a,b)\right) \right] \!\!\to\!\! \log_2\!\left( 1\!+\!P_{tot,a} \xi(a,b)\right)$ \cite{uniformly_integrable}.
    As $n\to \infty$, we can obtain \eqref{P_tot_conv_CSIR}.

    In addition, \eqref{P_tot_conv_CSIR_singleUE} is derived for a single UE sending $k$ bits with PUPE $\epsilon$ in quasi-static fading channels \cite{singleUE}.
  \end{proof}

\subsection{no-CSI} \label{section4_sub2}
  \begin{prop}\label{prop_converse_noCSI}
    We generate codebooks independently for UEs with each entry i.i.d. from $\mathcal{CN}(0,P)$. Given spectral efficiency $S$ and target PUPE $\epsilon$,
    we have $\varepsilon^{*} (M, \mu, p_a, \epsilon) \geq \inf \!\frac{P_{tot,a}}{S}$, where infimum is taken over all ${P_{tot,a}}\!>\!0$ satisfying
    \begin{equation}\label{P_tot_conv_noCSI}
      \ln \!M\;\!\!-\epsilon \;\!\!\ln (\;\!\!M\;\!\!-\;\!\!1)-h(\epsilon) \;\!\!\!\leq\;\!\!\!
      {M} \!\mathcal{V}\;\!\!\!\left(\!\frac{1}{{p_a}\;\!\!\mu M}\;\!\!, \!P_{tot,a}\!\!\right)\;\!\!-\;\!\!\;\!\!\mathcal{V}\;\!\!\!\left(\!\frac{1}{p_a\;\!\!\mu}, \!P_ {tot,a}\!\!\right) \!\;\!\!,
    \end{equation}
    \begin{equation}\label{P_tot_conv_noCSI_v}
      \mathcal{V}(r, \!\gamma)\!=\!r \!\ln (1\;\!\!+\;\!\!\gamma\;\!\!-\;\!\!\!\mathcal{F}(r, \!\gamma)\;\!\!)
      \;\!\!+\;\!\!\ln (1\;\!\!+\;\!\!r \gamma\;\!\!-\;\!\!\mathcal{F}(r,\! \gamma))\;\!\!-\;\!\!\frac{\mathcal{F}(r, \!\gamma)}{\gamma} \!,\!\!
    \end{equation}
    \begin{equation}\label{P_tot_conv_noCSI_f}
      \mathcal{F}(r, \gamma)\!=\!\frac{1}{4}\!\left(\!\sqrt{\!\gamma\!\left(\!\sqrt{r}+1\right)^{2}\!+\!1}
      \!-\!\sqrt{\!\gamma\left(\sqrt{r}-1\right)^{2}\!+\!1}\right)^{\!2}\!.
    \end{equation}
  \end{prop}

  \begin{proof}\label{proof_converse_noCSI}
    We assume a genie reveals the set of active UEs. Based on the analysis in Section \ref{section4_sub1}, we have
    \begin{equation}
      {K_a}\!\ln\! M \!-\! {K_a} \epsilon \ln(M\!-\!1) \!-\! {K_a} h\!\left(\epsilon\right)
      \!\leq \! I\!\left(\!\mathcal{X}_{\mathcal{K}_a}\!;\!\hat{\mathcal{X}}_{\mathcal{K}_a} \!\right)
      \!\leq\! I\!\left(\left.\bar{\boldsymbol{\beta}};\mathbf{y} \right|\!\bar{\mathbf{A}}\right),
    \end{equation}
    where $\bar{\boldsymbol{\beta}} \!\in \!\mathbb{C}^{K_{\;\!\!a} \;\!\!M}$ indicates which codewords are sent for active UEs and $\;\!\!\bar{\mathbf{A}}\;\!\!$ is the $n\times K_{\;\!\!a} \;\!\!M$ submatrix of $\;\!\!{\mathbf{A}}\;\!\!$ including codewords of active UEs.
    Let $\bar{\mathbf{H}}\!\in\!\mathbb{C}^{K_a M\times K_a M}$ be the submatrix of ${\mathbf{H}}$ including fading coefficients of active UEs.
    %where $\bar{\boldsymbol{\beta}} \!\in \!\mathbb{C}^{K_{\;\!\!a} \;\!\!M}$ and $\;\!\!\bar{\mathbf{A}}\;\!\!$ includes entries from ${\boldsymbol{\beta}}$ and $\mathbf{A}$ corresponding to actives UEs.
    Based on the chain rule of mutual information, we have
    \begin{align}
      I\left(\!\left.\bar{\boldsymbol{\beta}}, \bar{\mathbf{H}} \bar{\boldsymbol{\beta}}; \mathbf{y} \right|\!\bar{\mathbf{A}}\right)
      &\!=\! I\left(\left.\bar{\boldsymbol{\beta}};\mathbf{y} \right|\bar{\mathbf{A}}\right) \!+ \! I\left(\left.\bar{\mathbf{H}}\bar{\boldsymbol{\beta}};\mathbf{y} \right|\bar{\boldsymbol{\beta}}, \bar{\mathbf{A}}\right) \notag \\
      &\!=\! I\left(\left. \bar{\mathbf{H}}\bar{\boldsymbol{\beta}};\mathbf{y} \right|\bar{\mathbf{A}}\right) \!+\! I\left(\left.\bar{\boldsymbol{\beta}};\mathbf{y} \right|\bar{\mathbf{H}}\bar{\boldsymbol{\beta}}, \bar{\mathbf{A}}\right).
    \end{align}

    Since $\bar{\boldsymbol{\beta}}\!\to\! \bar{\mathbf{H}}\bar{\boldsymbol{\beta}}\!\to\! (\mathbf{y},\bar{\mathbf{A}})$ forms a Markov chain, the mutual information $I\!\left(\left.\bar{\boldsymbol{\beta}};\mathbf{y} \right|\!\bar{\mathbf{H}}\bar{\boldsymbol{\beta}}, \!\bar{\mathbf{A}}\right) \!=\! 0$. Hence, we have
    $ I\!\left(\left.\!\bar{\boldsymbol{\beta}};\mathbf{y} \right|\!\bar{\mathbf{A}}\right)
      \!=\! I\!\left(\left. \!\bar{\mathbf{H}}\bar{\boldsymbol{\beta}};\mathbf{y} \right|\!\bar{\mathbf{A}}\right) \!-\! I\!\left(\!\left.\bar{\mathbf{H}}\bar{\boldsymbol{\beta}};\mathbf{y} \right|\!\bar{\boldsymbol{\beta}}, \bar{\mathbf{A}}\right)$.
    We can obtain
    \begin{align}
      I\left(\left. \bar{\mathbf{H}}\bar{\boldsymbol{\beta}};\mathbf{y} \right|\bar{\mathbf{A}} = \bar{\mathbf{A}}_1\right)
      &=  I\left(\bar{\mathbf{H}}\bar{\boldsymbol{\beta}}; \bar{\mathbf{A}}_1\bar{\mathbf{H}}\bar{\boldsymbol{\beta}}+\mathbf{z} \right)\notag \\
      & \leq \sup_{\mathbf{u}} I\left(\mathbf{u}; \bar{\mathbf{A}}_1\mathbf{u}+\mathbf{z} \right) \notag \\
      &  = \ln \det\left( \mathbf{I}_{n} + \frac{1}{M}\bar{\mathbf{A}}_1\bar{\mathbf{A}}_1^{H} \right),
    \end{align}
    where $\bar{\mathbf{A}}_1$ is a realization of $\bar{\mathbf{A}}$ and the supremum is over random vector
    $\mathbf{u}$ with $\mathbb{E}\!\left[\mathbf{u} \right] \!\!=\! \mathbf{0}$ and $\mathbb{E}\!\left[\mathbf{u} \mathbf{u}^{H} \right] \!\!=\! \mathbb{E}\!\left[\!\left(\bar{\mathbf{H}}\bar{\boldsymbol{\beta}} \right)\!\! \left(\bar{\mathbf{H}}\bar{\boldsymbol{\beta}} \right)^{\!\!H} \right] \!=\!\! \frac{1}{M}\mathbf{I}_{K_aM}$.
    The supremum is achieved if $\mathbf{u}\!\sim\! \mathcal{CN}\!\left( \mathbf{0}, \!\frac{1}{M}\mathbf{I}_{K_aM} \right)$ \cite{sparsity_pattern}. Based on the random-matrix theory, we have \cite{random_matrix}
    \begin{equation}
      I\!\!\left(\!\left. \bar{\mathbf{H}}\bar{\boldsymbol{\beta}};\!\mathbf{y} \right|\!\bar{\mathbf{A}}\!\right)
      \!\!\leq\! \mathbb{E}\! \!\left[ \!\ln \!\det\!\left(\!\! \mathbf{I}_{n} \!\!+\!\!\frac{1}{M}\!\bar{\mathbf{A}}\bar{\mathbf{A}}^{\!\!H} \!\!\right) \!\right]
      \!\!\!=\!\!{K_a\!M} \mathcal{V}\!\left(\!\frac{1}{p_a\mu M}, \!P_{tot,a}\!\!\right)\!\!.
    \end{equation}

    For any realization $\bar{\mathbf{A}}_1$ of $\bar{\mathbf{A}}$ and $\bar{\boldsymbol{\beta}}_1$ of $\bar{\boldsymbol{\beta}}$, we have
    \begin{align}
      I\left(\left.\bar{\mathbf{H}}\bar{\boldsymbol{\beta}};\mathbf{y} \right|\bar{\boldsymbol{\beta}} \!=\! \bar{\boldsymbol{\beta}}_1, \bar{\mathbf{A}} \!=\! \bar{\mathbf{A}}_1 \right)
      &= I\left(\bar{\mathbf{H}}\bar{\boldsymbol{\beta}}_1; \bar{\mathbf{A}}_1\bar{\mathbf{H}}\bar{\boldsymbol{\beta}}_1+\mathbf{z}  \right) \notag \\
      &= I\left(\tilde{\mathbf{h}}; \tilde{\mathbf{A}}_1\tilde{\mathbf{h}} +\mathbf{z}  \right) \notag \\
      &= \ln \det\!\left( \mathbf{I}_{n} \!+\! \tilde{\mathbf{A}}_1 \tilde{\mathbf{A}}_1^{H} \right),
    \end{align}
    where $\tilde{\mathbf{h}} \in \mathbb{C}^{K_a}$ includes fading coefficients of active UEs and $\tilde{\mathbf{A}}_1$ is the $n\times K_a$ submatrix of $\bar{\mathbf{A}}_1$ formed by columns corresponding to the support of $\bar{\boldsymbol{\beta}}_1$. Hence, we have \cite{random_matrix}
    \begin{equation}
      I\!\left(\!\left.\bar{\mathbf{H}}\bar{\boldsymbol{\beta}};\mathbf{y} \right|\!\bar{\boldsymbol{\beta}}, \!\bar{\mathbf{A}}\right)
      \!=\! \mathbb{E}\!\left[ \ln \!\det\!\left( \mathbf{I}_{n} \!+\! \tilde{\mathbf{A}} \tilde{\mathbf{A}}^{\!\!H} \right) \!\right]
      \!\!=\! K_a \mathcal{V}\!\left(\frac{1}{p_a\mu}, \!P_ {tot,a}\!\right)\!.
    \end{equation}
  \end{proof}

\section{Results and Discussion} \label{section5}
  In this section, we evaluate the bounds derived in this work.

  Given the payload $k \!=\! 100$, active probability $p_a \!=\! 0.6$, and target PUPE $\epsilon \!=\! 0.001$, we show the trade-off of UE density $\mu$ with the minimum energy-per-bit $\varepsilon^{*}$ in Fig. 1.
  For TDMA, we split the blocklength $n$ equally among $K$ UEs.
  To achieve $S \!=\! p_a \mu k$, we obtain the smallest $P^{*}$ and $\varepsilon^{*} \!=\! P^{*}\!/(\mu k)$ ensuring the access of an active UE with rate $\mu k$, blocklength $1/\mu$ and PUPE $\epsilon$ based on the bound in \cite{TDMA_yangwei}.
  From Fig. 1, we can observe perfect multi-user interference (MUI) cancellation effect in quasi-static fading random access channels.
  It means that for small values of $\mu$, the optimal coding system can be performed as if each active UE is operated in isolation without interference.
  The orthogonalization scheme TDMA does not have this behavior.
  Although TDMA has better performance when $\mu\!\to\! 0$, it is
  more energy-inefficient at higher UE density.

  In Fig. 2, given $k \!=\! 100$ and $\epsilon \!=\! 0.001$, we show the trade-off of active UE density $p_a\mu$ and $\varepsilon^{*}$  with no-CSI.
  The converse bound and the achievability bound with the knowledge of UE activity are invariant for $p_a$ if $p_a\mu$ is fixed.
  Given $p_a\mu$, as $p_a$ increases, the uncertainty of active UEs decreases, and the achievability bound of random access can be reduced.
  Besides, this bound is close to the achievability bound with the knowledge of active UE set, %especially for small values of active UE density,
  i.e., only a bit more energy is required for random access compared with the case where UE activity is known.
  %with the knowledge of active UE set
  %knowing the active UE set
  \begin{figure}
  \centering
  \includegraphics[width=0.95\linewidth]{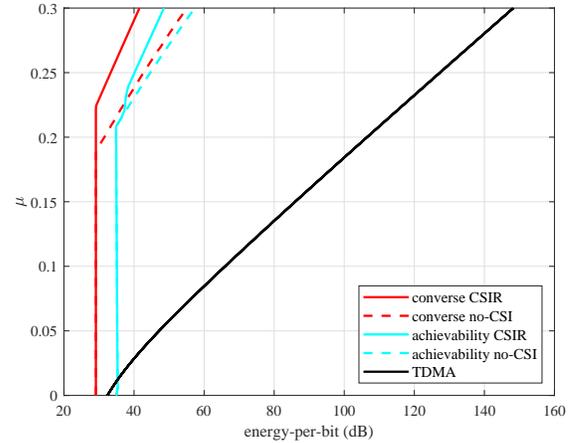}\\
  \caption{$\mu$ versus $\varepsilon^{*}$ with $k = 100$, $p_a = 0.6$, and $\epsilon = 0.001$.}
  \label{fig:1}
  \end{figure}
  \begin{figure}
  \centering
  \includegraphics[width=0.95\linewidth]{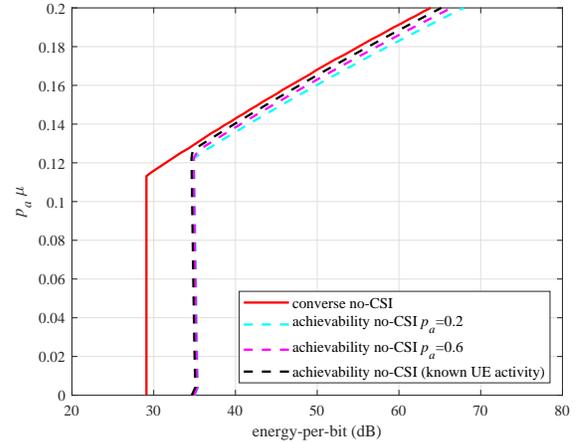}\\
  \caption{$p_a\mu$ versus $\varepsilon^{*}$ with $k = 100$ and $\epsilon = 0.001$.}
  \label{fig:2}
\end{figure}
\section{Conclusion} \label{section6}
  In this work, we assume the number of UEs grows linearly and unboundedly with the blocklength and each active UE has finite data bits to send over quasi-static fading channels. The achievability and converse bounds on the minimum energy-per-bit for reliable massive random access with CSIR and no-CSI are derived.
  Simulation results show perfect MUI cancellation for small values of UE density $\mu$ and the energy-inefficiency of TDMA as $\mu$ increases.
  Besides, with no-CSI, the energy-per-bit for random access is only a bit more than that with the knowledge of UE activity.
  The bounds derived in this work provide energy-efficiency targets for future massive random access coding and communication schemes.

\end{document}